\documentclass[11pt,a4paper,final]{article}
\usepackage[left=1in, right=1in, top=1in, bottom=1in]{geometry}
\usepackage{booktabs}
\usepackage{times}
\usepackage{tikz}
\usepackage{lineno,hyperref}
\usepackage[utf8]{inputenc}
\usepackage[T1]{fontenc}
\usepackage[english]{babel}
\usepackage{amsmath}
\usepackage{amsfonts}
\usepackage{amssymb}
\usepackage{makeidx}
\usepackage{tfrupee}
\usepackage{graphicx}
\usepackage{ifvtex}
\usepackage{setspace}
\usepackage{float}
\usepackage{color,hyperref}
\usepackage{fancyhdr,graphicx,amsmath,amssymb}
\usepackage[ruled,vlined]{algorithm2e}
\usepackage{algpseudocode}
\usepackage{algorithmicx}
\usetikzlibrary{backgrounds}
\usetikzlibrary{arrows}
\usetikzlibrary{shapes,shapes.geometric,shapes.misc}
\pgfdeclarelayer{edgelayer}
\pgfdeclarelayer{nodelayer}
\pgfsetlayers{background,edgelayer,nodelayer,main}
\tikzstyle{none}=[inner sep=0mm]
\tikzset{new style 0/.style={circle,draw}}

\newtheorem{lemma}{Lemma}
\newtheorem{proposition}{Proposition}

\newtheorem{definition}{Definition}
\newenvironment{proof}{\paragraph{Proof:}}{\hfill$\square$}
\title{A Model of Competitive Assortment Planning Algorithm}
\author{Dipankar Das\footnote{Assistant Professor, Goa Institute of Management, Email ID:dipankar3das@gmail.com;dipankar@gim.ac.in}}
\begin{document}
	\maketitle	
	\begin{abstract}
With a novel search algorithm or assortment planning or assortment optimization algorithm that takes into account a Bayesian approach to information updating and two-stage assortment optimization techniques, the current research provides a novel concept of competitiveness in the digital marketplace. Via the search algorithm, there is competition between the platform, vendors, and private brands of the platform. The current paper suggests a model and discusses how competition and collusion arise in the digital marketplace through assortment planning or assortment optimization algorithm. Furthermore, it suggests a model of an assortment algorithm free from collusion between the platform and the large vendors. The paper's major conclusions are that collusive assortment may raise a product's purchase likelihood but fail to maximize expected revenue.  The proposed assortment planning, on the other hand, maintains competitiveness while maximizing expected revenue.
\end{abstract}
\section{Introduction}
Assortment personalization is the problem of selecting the best assortment of products for each customer to maximize revenue. This is a critical issue in e-commerce revenue management. Customers select which products to purchase from the assortment presented by the Platform using the customized search algorithm. Platforms decide which assortment to present to the customer and collect revenue based on each customer's selection.\\
A central problem in revenue management is selecting the assortment that maximizes expected revenue. This is known as assortment planning or assortment optimization. \\
There is a direct relationship between online sequential decision-making and reviews, assortment planning, and the revenue management of the platform. Significant research found in this regard in a set of recent articles. These are \cite{Besbes_2009, doi:10.1287/opre.2016.1534, doi:10.1287/isre.2019.0852, doi:10.1287/mnsc.2019.3372, doi:10.1287/mnsc.2022.4387, doi:10.1287/mnsc.2021.4044, Chen_2021, doi:10.1287/opre.2022.2380, Vaccari_2018}. The main research gaps are as follows. The revenue maximization strategy in the assortment planning problem has not been considered the competition in the platform. Rather focuses on the self-revenue maximization problem and promotes its brands despite having an insignificant number of reviewers\footnote{Market study on e-commerce in India(08-01-2020):\url{https://www.cci.gov.in/node/4637}}. The consumer has to consider the reviewers as rational in the decision-making under the Bayesian setup.  Hence, the present article tries to reduce the gap and proposes a formal model of assortment planning considering the competition between the vendors and platform using-adjusted average number of reviews in the customers' decision-making and the product ranking.\\ 
 A central problem in revenue management is selecting the assortment that maximizes expected revenue. This is known as assortment planning or assortment optimization. 
The current article identifies that the existing search algorithms used in the platform's information retrieval are not based on standard competitive models that meet all of the criteria. Problems persist in digital markets, such as unequal market access on the part of small vendors on the platform, the creation of information asymmetry in terms of the availability of relatively competitive vendors to customers, the dominance of large vendors in gaining extra market access in the digital market, and so on. As a result, the study presents a standard assortment planning or assortment optimization technique that avoids collusion between the platform and large vendors, provides equal access, decreases information asymmetry, and maximizes expected revenue.\\ 
The online platforms serve as both a marketplace and a competitor in that marketplace, and they have the incentive to use their platform control to benefit their own/preferred vendors or private label products at the expense of other sellers/service providers on the platform. The platform's intermediary role allows it to gather all such competitively relevant data as price, sold quantities, demand, and so on for each product, seller, and geography. On the customer side, this allows the platform to better target product recommendations for users and improve platform quality. On the seller's side, this may allow it to use such data to introduce its private label or boost its or its 'preferred sellers' sales. Organic search ranking is generated by the platform's search algorithm, and thus the platform controls the search parameters and results. The platform's dual role raises the possibility of ranking biases created by the platform as a discriminatory device. The search algorithm's 'Black Box' nature limited customers' ability to identify biases, limiting the possibility of self-correction. The platforms identified user review/rating as a critical input in determining search ranking. Lack of transparency and credibility issues surrounding some of these platforms' user review and rating policies were highlighted by business users as a factor that further allowed for search result manipulation, affecting their ability to compete effectively with vertically integrated entities or the platforms' preferred entities. It is not in their best interests to prioritize a few at the expense of platform quality, especially given the intensity of competition between online platforms and the still predominantly offline retail landscape. The issue of platform intensity and lack of countervailing power is dependent on the large platforms, as evidenced by their inability to go offline individually even when they find the contract terms imposed by the platforms to be abusive and unacceptable. The above argument is supported by a significant empirical study conducted by the Competition Commission of India \footnote{Market study on e-commerce in India(08-01-2020):\url{https://www.cci.gov.in/node/4637}}. Concerns were raised about the lack of clarity on search ranking criteria across all three categories of work. The main issue is that when the search query matches, local brands rank lower than the platform's private brands. The current paper presents a method of collusion free and competitive assortment algorithms by employing a base model of information acquisition via sequential online review data that maximizes revenue. It shows that the collusive assortment algorithm may increases the purchase probability of a particular product but fails to maximize revenue. The primary distinction between search engines such as Google, Yahoo, and others and the search engine used by Amazon, Flipkart, and others is that the Google-type search engine allocates slots based on auction bidding, with each bidder paying for the allotted slots. A set of related literature can be found in \cite{varian2007position,varian2009online,edelman2007internet,haucap2014google,banchio2022adaptive}.
On the other hand, a search engine of the Amazon type (or e-commerce platform) allocates slots using both with and without requiring any payment for the allotted slot for all vendors.  Initial slots are allocated using an auction mechanism but the rest slots are allocated using a score and the score is based on the individual firms’ performances. The present article proposes a model for the slots that are being allocated using a scoring mechanism. Individual firms that are listed have a built-in score. If the search query matches a set of listed firms on a given platform, the firms' scores will be permitted. Platform, on the other hand, adheres to a joint profit maximization philosophy. The platform's profit-maximizing slot allocation is in place. If these two preferences match, there will be no problem and the matching will be stable; however, if the platform's allocation does not match the allocation of the firm's score, there will be a mismatch and the matching will be unstable. This mismatch is regarded as collusion in the sellers' ranking. As a result, in addition to their score, the platform will have a mechanism for transferring utility to firms that did not get the desired slots. Choice can be influenced by the position of alternatives, as cited in \cite{brook1974biases}. There is a substantial body of literature on the rank order of alternatives. \cite{critchlow1991probability,georgescu1958threshold,georgescu1969relation,luce1977choice,plackett1968random,plackett1975analysis} are some examples.  The current article proposes a theoretical model for it by describing the problem of competition and collusion through the design of an assortment mechanism. Moreover, the present article proposes a model that demonstrates how a collusive assortment algorithm generates a collusive ranking and biases in choice.
\subsection{Related Literature}
Nowadays, digital platforms serve two functions: they run marketplaces for third-party products while also selling their products on those marketplaces. The company that owns and manages the marketplace system will also sell competing products through it. This results in a new type of channel conflict, \cite{ryan2012competition}. According to a recent study, a ban often benefits third-party sellers at the expense of customer surplus or overall welfare. The main reason for this is that the presence of the platform's products constrains the pricing of third-party sellers on its marketplace in dual mode, which benefits customers  \cite{hagiu2022should}. \cite{mills1995retailers,gamp2022competition,tian2018marketplace} for a collection of related literature on third-party retailer models.
In the United States, Amazon is used by far more customers than Google for product searches\footnote{Statista, July 7, 2021}. The ranking of the sellers is determined by three distinct sets of agents: the customer's query, the sellers' performance parameters, and the platform's profit-maximizing objective. Technically, the primary goal of product ranking is to match customers' preferences with existing products that match the match.  \textit{Ranking in Information Retrieval }  is a method of ranking the objects related to the customer's query. \cite{liu2009learning} contains a collection of significant research.
A relevance ranking model attempts to generate a ranked list of documents based on their relevance to the query. These information retrieval methods attempt to minimize the differences between the estimated and actual query outputs. Another approach to sequential prediction problems is \textit{Prediction, Learning, and Games}, and significant literature can be found in \cite{cesa2006prediction}. There is evidence of control over the search engine platform and the creation of a market flaw. This means that the platform can also manipulate the product ranking. According to \cite{argenton2012search}, there is a strong tendency toward market tipping and, as a result, monopolization, with negative consequences for economic welfare.
The leader has the incentive to exploit them and attract more customers to monopolize advertising. In terms of entry barriers, the leading platform has a strategic incentive to exploit scale in search, manipulate search results to divert search traffic away from competitors, limit multi-homing, expand its market share, and deny scale to competitors. \ cite{etro2013advertising}. In \cite{evans2013antitrust,dolata2017apple,chiou2017search}, a survey of the economics literature on multi-sided platforms was conducted, with a particular focus on competition policy issues such as market definition, mergers, monopolization, and coordinated behavior.
The \cite{sorokina2016amazon,karmaker2017application} surveyed e-commerce ranking methods. However, there is no research on the formation of collusion with the seller.\\
The platform's dual-Format model is critical in forming the collusion. Amazon, the world's largest online retailer, is the most well-known example of dual-format retailing. Amazon Marketplace, in addition to performing the traditional "merchant" function (buying and reselling goods), provides a platform for 3P sellers to compete for the same customers. A customer shopping for a product (for example, a camera) has the option of purchasing it directly from Amazon or through other Amazon Marketplace sellers in many product categories (as a new or used product). A retailer's strategic rationale for introducing a 3P marketplace is that it provides an "outside option" that improves its bargaining position in negotiations with the manufacturer \cite{mantin2014strategic}.  This can explain the growing popularity of such marketplaces, which provide a platform for directly connecting sellers and buyers. The intermediary acts as an agent in the online marketplace mode, charging a proportional fee for each sale. A fee of this type could be viewed as a revenue-sharing mechanism that could reduce double marginalization\cite{tian2018marketplace}.\\
Several recent studies have been published to address the question of whether platforms should be allowed to sell on their marketplaces. The relevant literature can be found in \cite{hagiu2022should};
\cite{gamp2022competition}; \cite{sato2022joint}. The following are the key findings. Even when focusing on the same narrow product category, the dual mode has advantages. A prohibition on this dual mode frequently favors third-party sellers at the expense of customer surplus or overall welfare. The main reason for this is that the presence of the platform's products constrains the pricing of third-party sellers on its marketplace in dual mode, which benefits customers. In the dual model and joint purchase mode, on the other hand, the platform can profitably increase the commission by increasing cross-market complementarity. Firms may design low-quality products to market them to gullible customers who misjudge product characteristics. As a result, superior and inferior quality coexist, and as search frictions disappear, the share of superior goods approaches zero.\\
From the preferences of the customers, a set of interesting literature has been found. A fundamental problem in click data is position bias. The probability of a document being clicked depends not only on its relevance but on its position on the results page \cite{craswell2008experimental, kveton2015cascading,zong2016cascading}. \\
Assortment personalization is the problem of choosing, for each individual (type), the best assortment of products, ads, or other offerings (items) to maximize revenue. This problem is central to revenue management in e-commerce and online advertising where both items and types can number in the millions. A set of related literature is found in \cite{kallus2020dynamic,agrawal2019mnl,abeliuk2016assortment,ferreira2022learning,gao2021assortment,liu2020assortment,saure2013optimal}\\\cite{wang2018impact}.  There are important articles on assortment planning and optimization model viz. \cite{chen2021revenue,doi:10.1287/mnsc.2022.4387,vaccari2018social}. These three articles have considered three different approaches one considered a static problem, another consider a dynamic approach but did not consider the competitiveness in the dynamic phenomena.\\
There is an important relationship between product ranking and revenue management. The existing literature considers either maximizing revenue by preferring the probability of satisfaction with any product. The paper \cite{doi:10.1287/opre.2016.1534} proposes Rank Centrality, an iterative rank aggregation algorithm for discovering scores for objects (or items) from pairwise comparisons. The article \cite{derakhshan2022product}  proposes a two-stage sequential search model wherein the first stage, the customer sequentially screens positions to observe the preference weight of the products placed in them and forms a consideration set. In the second stage, she observes the additional idiosyncratic utility that she can derive from each product and chooses the highest-utility product within her consideration set. A revenue maximization problem for an online retailer who plans to display in order a set of products differing in their prices and qualities in \cite{chen2021revenue}. The product ranking mechanisms of a monopolistic online platform in the presence of social learning has been considered in \cite{vaccari2018social}. Other related research works are found in \cite{agarwal2011location,besbes2012blind,talluri2004revenue}. The present work considered the justified competitive product ranking and revenue maximization strategy under the dynamic phenomena. \\
The main draw of this article is how the platform's colluding behavior affects the allocation of slots or the ranking of the objects, purchase probability, and revenue. Under collusion, the platform's main goal is to change the order of the preferred shortlisted objects in \textit{Add to Cart}. As a result, competition will be reduced, and buyers will be directed toward predefined objects/products to prefer. As a result, this is a mechanism design problem in which the platform will first set the result and then the arrangements will be in place to achieve the predefined results.
\section{Statement of the Problem and Summary of Findings}
This section covers the suggested model and findings of the two-stage assortment method. Figures \ref{figure3} and \ref{figure4}  illustrate two examples of search quarries. The first four ranks are based on a payment basis, or an auction system, as shown in Fig. \ref{figure4}. This is labeled as sponsored. However, from rank five onwards, the list is based on a scoring system.  It is clear that the product/firm/vendor number eight is not appropriate for the $ 8^{th} $ rank position. It is evident that with the same review ratings, the number of reviewers is $ 95 $, and the price is likewise lower. As a result, the article demonstrates that the purchase probability of the $ 9th $ ranked vendor's products will increase, but not the expected revenue. Using Fig. \ref{figure4} we have prepared two tables Table \ref{tab:title} \& Table \ref{tab:title continue} to select the first three-ranked objects using the proposed two stage assortment algorithm. The First Stage Ranking has been done using the minimum cut-off review numbers and quality as the parameter.  The first stage cut-off has been calculated by dividing the summation of column (6) by the summation of column (3). Thereafter, the objects will be selected that have these minimum cut-off review numbers and arranging them concerning average review ratings from highest to lowest. Thereafter in Stage Two, the cut-off has been revised concerning price. In stage two, the cut-off review numbers will be calculated by dividing the summation of column (4) by the summation of column (4) on the shorted objects in stage one. Thereafter, the objects will be arranged in ascending order concerning price from highest to lowest. And finally, we got the first ranked object as A. We may have got more than one object for the subsequent positions. As we have got only one object here so, the next places will be calculated using the same way by eliminating object A. The second-ranked object will be B. And the third object will be F.  Here it has been assumed that the customer has a fixed attention span i.e. 3. Hence, we have compared the expected revenue by comparing with another collusive ranking i.e. A-D-F. The article has proved in Proposition 1,2,\&3 that the purchase probability for product F in the second list will be more than the actual one. i.e. A-B-F. But this will not guarantee that the expected revenue will be maximum. Because the price of product D may be lower than product B because the demand is lower for product D compared to B. This has been proved in the article and especially in Proposition 1, 2, and 3. 
 \begin{figure}[H]
	\centering\includegraphics[scale=0.40]{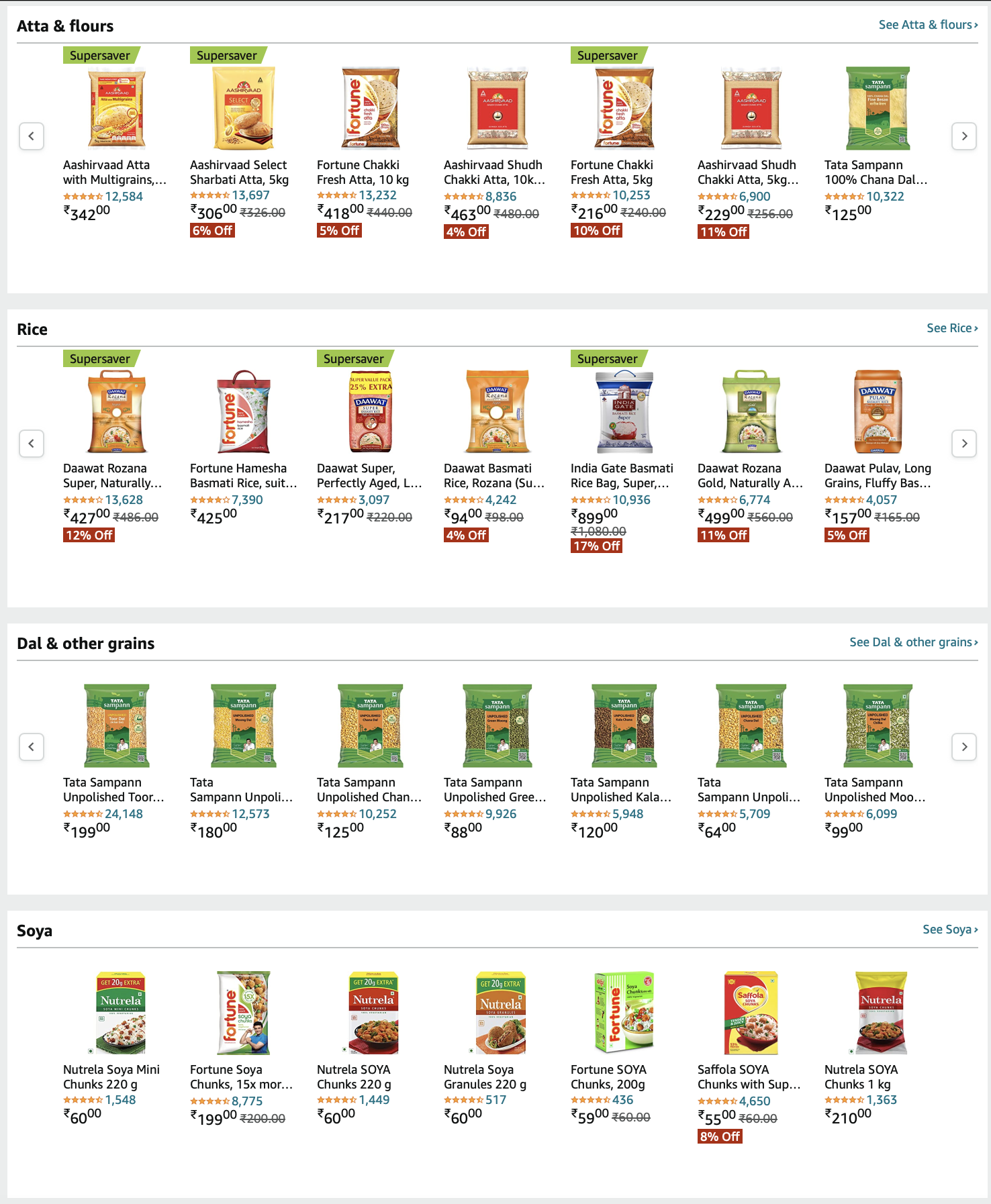}\caption{Product Rank of Food Items-Example from amazon.in}\label{figure3}
\end{figure} 
 \begin{figure}[H]
	\centering\includegraphics[scale=0.40]{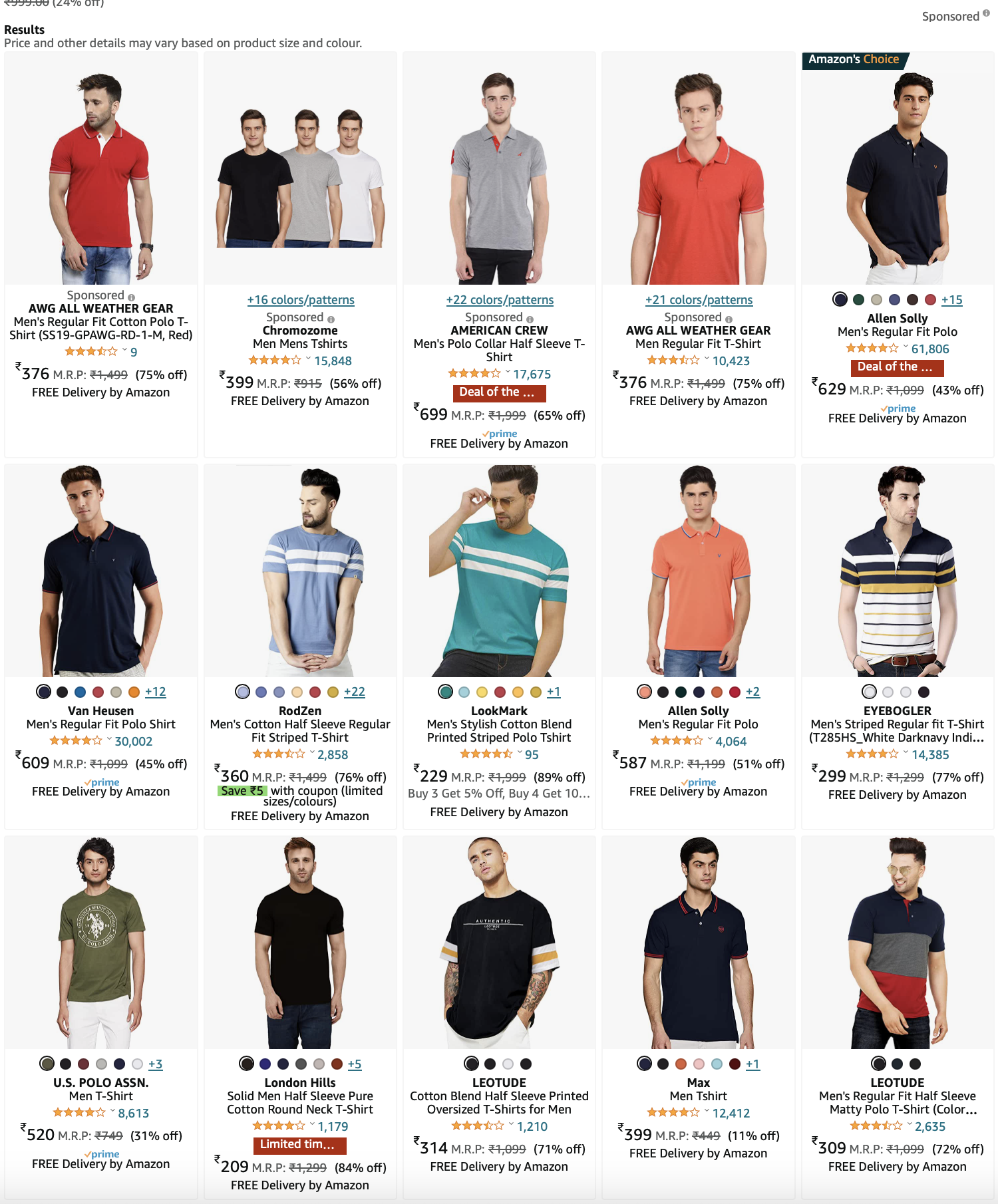}\caption{Product Rank of T-Shirt for Men-Example from amazon.in}\label{figure4}
\end{figure} 
\begin{table}[H]
	\caption {Two Stage Assortment Planning Algorithm} \label{tab:title}
	\begin{center}
	\begin{tabular}{@{}|l|l|l|l|l|@{}}
		\toprule
		\textbf{\begin{tabular}[c]{@{}l@{}}Product Name\\        (1)\end{tabular}} & \textbf{\begin{tabular}[c]{@{}l@{}}Number of Reviews               \\             (2)\end{tabular}} & \textbf{\begin{tabular}[c]{@{}l@{}}Average Ratings             \\           (3)\end{tabular}} & \textbf{\begin{tabular}[c]{@{}l@{}}Price  \\      (4)\end{tabular}} & \textbf{\begin{tabular}[c]{@{}l@{}}Purchase Probability \\               (5)\end{tabular}} \\ \midrule
		A                                                                          & 61806                                                                                               & 4                                                                                             & 629                                                                 & 0.95                                                                                       \\ \midrule
		B                                                                          & 30002                                                                                               & 4                                                                                             & 700                                                                 & 0.85                                                                                       \\ \midrule
		C                                                                          & 2858                                                                                                & 3.5                                                                                           & 360                                                                 & 0.40                                                                                       \\ \midrule
		D                                                                          & 95                                                                                                  & 4.5                                                                                           & 229                                                                 & 0.10                                                                                       \\ \midrule
		E                                                                          & 4064                                                                                                & 4                                                                                             & 587                                                                 & 0.55                                                                                       \\ \midrule
		F                                                                          & 14385                                                                                               & 5                                                                                             & 299                                                                 & 0.75                                                                                       \\ \midrule
		G                                                                          & 8613                                                                                                & 4                                                                                             & 520                                                                 & 0.65                                                                                       \\ \midrule
		H                                                                          & 1179                                                                                                & 4                                                                                             & 209                                                                 & 0.20                                                                                       \\ \midrule
		I                                                                          & 1210                                                                                                & 3                                                                                             & 314                                                                 & 0.15                                                                                       \\ \midrule
		J                                                                          & 12412                                                                                               & 4                                                                                             & 399                                                                 & 0.72                                                                                       \\ \bottomrule
	\end{tabular}
\end{center}
\end{table}

\begin{table}[H]
	\caption {Two Stage Assortment Planning Algorithm-Continue} \label{tab:title continue}
	\begin{center}
	\begin{tabular}{|l|l|l|l|l|}
		\hline
		\textbf{(6)=(2)X(3)} & \textbf{(7)=(2)X(4)} & \textbf{\begin{tabular}[c]{@{}l@{}}First Stage Ranking\\                (8)\end{tabular}} & \textbf{\begin{tabular}[c]{@{}l@{}}Second Stage Ranking\\                   (9)\end{tabular}} & \multicolumn{1}{c|}{\textbf{\begin{tabular}[c]{@{}c@{}}Final Rank\\ (10)\end{tabular}}} \\ \hline
		247224               & 38875974             & F                                                                                         & A                                                                                             & A                                                                                       \\ \hline
		120008               & 21001400             & A                                                                                         & B                                                                                             &                                                                                         \\ \hline
		100003               & 1028880              & B                                                                                         & F                                                                                             &                                                                                         \\ \hline
		427.5                & 21755                &                                                                                           &                                                                                               &                                                                                         \\ \hline
		16256                & 2385568              &                                                                                           &                                                                                               &                                                                                         \\ \hline
		71925                & 4301115              &                                                                                           &                                                                                               &                                                                                         \\ \hline
		34452                & 4478760              &                                                                                           &                                                                                               &                                                                                         \\ \hline
		4716                 & 246411               &                                                                                           &                                                                                               &                                                                                         \\ \hline
		3630                 & 379940               &                                                                                           &                                                                                               &                                                                                         \\ \hline
		49648                & 4952388              &                                                                                           &                                                                                               &                                                                                         \\ \hline
	\end{tabular}
\end{center}
\end{table}
\section{The Model: Competitive Assortment Algorithm}
A recent article studies two objectives: (i) maximizing the platform’s market share and (ii) maximizing the customer’s welfare \cite{derakhshan2022product}. When a customer (she) arrives, she views the products sequentially. However, the retailer cannot perfectly predict how many products she is willing to view. customers have attention spans, i.e., the maximum number of products they are willing to view and inspect the products sequentially before purchasing a product or leaving the platform empty-handed when the attention span gets exhausted. The customer seems to be more likely to buy a product ranked at the top, even though there is another similar product below inside her attention span. In this connection, an important question has been solved on how the ranking is linked to the revenue maximizing\cite{chen2021revenue}. A competitive search algorithm must fulfill the two important criteria. First is that each non-sponsored vendor should have equal market access and second each customer should have been exposed to all the possible low-price and quality alternatives during the search process. What exactly is a competitive search algorithm? A competitive search algorithm solves problems that have two sides. The first is the problem of the traders/vendors and the platform's incentive, and the second is the problem of the customer. When a search query is entered, the search algorithm should prioritize the best vendors who meet the query and the criteria for being considered a good vendor. Second, the search algorithm should take into account the customer's best matching vendors among the above-mentioned vendors.\\
This is based on the theory that a group of firms (vendors) will compete with one another and eventually compete in the market, based on their strengths and quality. These vendors are the best firms in the market that are available to customers. The customer will choose vendors who minimize mismatch and thus solve the customer's choice problem. The section that follows will help you understand the customer's mismatch problem.
\subsection{Vendors Ranking Problem}
Before we begin deriving the algorithm, let us first understand the problem of vendor shorting and then the problem of mismatch of the customer's choice set using examples. First-level shortlisting is the problem of assigning each indivisible vendor to each position, which will be done in terms of ranking. And this ranking will be based on the inbuilt score, which is based on a set of criteria, such as review ratings, comments, price, a product in high demand, and so on... This must be done by all of the parameters.  This has already been explained in the preceding sections. The model of assigning an equal number of indivisible objects to an equal number of places based on their scores is the first-level ranking. This has been explained using Shapley and Scarf's model \cite{SHAPLEY197423}, which considers each criterion as an agent influencing the assignment of each object to each place and derives the final assignment. The prescribed collusive ranking model in this article demonstrates that if the platform acts as an agent and redefines a new assignment, the final assignment will be changed to a collusive or joint profit-maximizing assignment based on the criteria. It demonstrates that if any of the following criteria are proposed, namely \textit{products related to this item}, \textit{more to consider from made for Amazon},\textit{compare with similar items},\textit{customers who viewed this item also viewed}, etc...  then the allocation will be collusive. By minimizing the customer's mismatch problem, the next level ranking or shortlisting vendors is to be done out of the first level shortlisting vendors. This mismatch can be reduced by updating online reviews with a Bayesian approach. Because online reviews are nothing more than sequential voting, they contain a history of customer preferences. As a result, relying solely on the current review comments is insufficient to assess the object's quality. This information must also be updated to determine the quality of future services.
\subsection{Model Setup: Problem Formulation}
Imitation of third-party products and self-preferencing are important tools used to direct customers to the platform's products\cite{hagiu2022should}. As a result, it should maintain equal access from both sides, namely the customer and the firms. On E-Commerce platforms, there are currently two types of seller ranking. The first is sponsored, while the second is not. Sellers who chose the sponsored ranking mechanism must participate in the auction bidding process to win the slot. The remainder is managed by manipulating keywords, review ratings, review comments, and so on. Furthermore, platforms have some hidden mechanisms in place to push their products. The current article proposes a general model for retrieving the query so that every seller can benefit concerning their strength, particularly those who do not participate in the sponsored ranking mechanism. The proposed model is the development of an algorithm that solves the preceding two problems so that the search algorithm is not collusive. Otherwise, the search algorithm will be collusive, and the platform will have the power to control the market in favor of certain interest groups.
\subsubsection{Analysis: Assortment Planning and the Revenue Maximization Problem}
We consider a marketplace (platform) where a set of $ K $ substitutable goods or services offered by the $ K $ individual vendors, henceforth, called the \textit{products}-are offered to a market of customers who decide whether to buy one of them or to choose a no-purchase option. Let $ [k]\triangleq \{1,...,k\} $ denote the universal set of products. The online platform (he) chooses an assortment $ S\subset [k] $ to display and rank them in order within $ M  $ slots.  This has to be done in tow stage process; $ S_{1}\& S_{2}(=S) $. First $ S_{1}\subset [k] $ thereafter $ S_{2}(=S) \subseteq S_{1}$. The platform sells this single product over a horizon of length $ T $ and each vendor is holding more than one product. It will be convenient to think of this planning horizon in terms of customer arrivals. Consider, customers arrive sequentially and are indexed by $ t\in \{1,..., T\}, $ and it will be assumed henceforth that $ T $ is known to the platform. customers are indexed by $n=1,2,...$. Each customer purchases at most one product upon arrival and gives review comments about the product in the post-purchase situation. It does not buy her returns to the market without giving the review comments. The platform can influence demand by varying the price it offers to different customers and by altering the product rank concerning time $ t $. The product’s price history is not observable to customers, which is suggestive of the practice in most online marketplaces. The demand also depends on the product’s quality, which is unknown to customers upon arrival. Reviews reported by previous buyers provide public information about the product, and customers use this information to estimate the product’s quality. In what follows, we formalize the functional relationships among demand, price, and review outcomes. $q_{k}$  represents the intrinsic quality of product $ k $ and $ p_{k} $ is the price of the product. The index $ k=0 $ indicates the no-purchase option, which has known intrinsic quality $ q_{0}=0 $ and a price $ p_{0}=0 $.
\paragraph{Product rankings and search costs.} The search cost associated with a given product depends on the position in which the product is displayed to customers. To formalize the search cost as a function of product ranking, we let $ \Gamma^{K} $ denote the set of all permutations of $ \{1,..., K\} $ and let $ \triangle(\Gamma^{K}) $ represent the space of all probability distributions over $\Gamma^{K} $. Elements of $ \Gamma^{K} $ will be referred to as position assignments or, more simply, as rankings. Given a product ranking $\mathbf{\gamma}=\{ \gamma_{1},...,\gamma_{K} \}\in  \Gamma^{K}$, $ \gamma_{k}=j $ indicates that product $ k $ occupies the $ j^{th} $ highest position in the ranking. For instance, when $ \gamma_{k}=1 (\gamma_{k}=K)$, product $ k $ occupies the highest (lowest) position. In particular, if a product occupies the $ j^{th} $ highest position in the ranking, customers incur a search cost $ g(j) $, where $ g:\mathbb{N}\rightarrow \mathbb{R} $ is a strictly increasing function; without loss of generality, we normalize $ g(1)=0 $.
\paragraph{Customers' (she) attention span.} After arriving on the platform the customer views the products sequentially. At this moment it is difficult to predict by the platform that how many products she is willing to buy and which products. Let it is captured by the random attention span $ Y $, with distribution $ h_{y}\triangleq Pr(Y=y) $ and $ H_{y}\triangleq Pr(Y\geq y)$ for $ y=1,2,... $. For a customer with an attention span of $ y $ drawn from $ Y $, she views the first product and buys it if it is satisfactory, and so on. And let when it is captured by the fixed attention span then it is $ y $.
\paragraph{Customers' valuation.} I assume that the value of a product $ k $ to customer $ n $  at time $ t $ is defined by $ x_{n,k,t}=\beta-\alpha p_{n,k,t}+ q_{n,k,t}-g(\gamma_{k}) $
 where $ p_{n,k,t}\in[\underline{p},\bar{p}] $, with $ \underline{p}<\bar{p} $ is the price quoted by the platform to customer $ n $ for product $ k $ at time $ t $. 
The intercept $ \beta $ reflects the product’s value that is common to customers and the slope $ \alpha $ represents customers’ sensitivity to price. The value of $ q_{t} $ represents the quality experienced by customer $ n $ after purchase, which is taken to be an independent normal random variable with mean $ \mu $ and variance $ \sigma^{2} $. The parametric structure of the underlying distribution for $ q_{n,k,t} $ is known to customers except for its mean $ \mu $. I assume that the value of $ \sigma $ is known for simplicity in exposition and transparent analysis. The case with unknown $ \sigma $ can be incorporated into our model without changing our main insights, albeit with some technical difficulties associated with the Bayesian inference with unknown variance. We assume that the seller has full information on the distribution of $ q_{n,k,t} $. Without loss of generality, we normalize $ (\alpha, \beta) = (1, 0) $ in our theoretical analysis so that the valuation of customer $ n $ is written as $ x_{n,k,t}=q_{n,k,t}-p_{n,k,t}-g(\gamma_{k}) $ Note that the value of $ q_{n,k,t} $, and hence $ x_{n,k,t} $, is unknown to customer $ t $ before purchase. Given a vector of quality estimates $  \mathbf{\hat{q}_{n,k,t}} $, and a position assignment $ \mathbf{\gamma}\in\Gamma^{K} $, customer $ n $ assigns a utility $ \hat{q}_{n,k,t}-p_{k}-g(\gamma_{k}) $ to the purchase of product $ k $ and a utility of $ \hat{q}_{n,0,t}-p_{0}-g(\gamma_{0})=0 $ to the outside option.
\paragraph{Customers' purchase decisions.}Then she buys the product by maximizing the expected utility i.e. $ \arg \max_{k=0,1,..., K} E[\{\hat{q}_{n,k,t}-p_{k}-g(\gamma_{k})\}]$, where $ g(\gamma_{0})=0 $. Customer $ n $ observes some information $ \Phi $ on past reviews before making the purchasing decision. The expected utility $ \chi_{n,k,t}$ of a purchase given $ \Phi_{n,k,t} $ can be written as
\begin{equation*}\label{1}
	\chi_{n,k,t}=E[\mu\mid \Phi_{n,k,t}]-p_{n,k,t}-g(\gamma_{k})
\end{equation*}
The probability of purchase follows the familiar logit model, hereafter referred to as the demand function, denoted by $ \lambda (p_{t},\Phi_{t}, \gamma)$ where
\begin{equation}\label{2}
	\lambda (p_{t},\Phi_{t}, \gamma) =\dfrac{\exp(\chi_{n,k,t})}{1+\exp(\chi_{n,k,t})}= \dfrac{\exp(E(\mu\mid\Phi_{t})-p_{t}-g(\gamma_{k})}{1+\exp(E(\mu\mid\Phi_{t}-p_{t}-g(\gamma_{k}))}
\end{equation}
\paragraph{Customers’ belief updating rules.}  The belief updating rules are based on \cite{shin2022dynamic}. To be specific, if customer $ n $ purchases the product, she reports a rating of $ q_{n,t} $. Then, the dynamics
of the review platform is described by $ \Phi_{n,k,t}=(n_{t},\bar{q}_{t}) $ where $ n_{t} $, is the number of reviews and $ \bar{q}_{t}=\frac{\sum_{i=1}^{n_{t}}q_{t(i)}}{n_{t}} $ for $ n\geq 2 $ is the average rating, and where $ t(i) $ is the time index of the $ i^{th} $ reviewer. For $ t=1, $ we let $ n_{1}=\bar{q}_{1}=0 $. The state of the review platform is updated as follows:
\begin{equation*}\label{3}
	(n_{t+1},\bar{q}_{t+1}) = 	\left\{
	\begin{array}{ll}
	(n_{t}+1,\frac{n_{t}\bar{q}_{t}+q_{n_{t+1}}}{(n_{t}+1)}) & \mbox{if customer $ n $ purchases,}  \\
	(n_{t},\bar{q}_{t}) & \mbox{otherwise}
	\end{array}
	\right.
\end{equation*}
Customers share a common prior belief, expressed in our model through a normal random variable $ \bar{q}_{0}\sim N(\mu_{0},\sigma_{0}^{2})$. After observing the state $ \Phi $ of the review platform, customer $ t $ updates her belief according to Bayes' rule. The mean of the posterior belief denoted by $ E[\mu\mid \Phi_{t}] $ as 
\begin{equation}\label{4}
	\mu_{t}=\dfrac{1}{\rho n_{t}+1}\mu_{0}+\dfrac{\rho n_{t}}{\rho n_{t}+1}\bar{q}_{t}
\end{equation}
and $ \rho=\dfrac{\sigma_{0}^{2}}{\sigma^{2}} $. The above formula and the derivations are taken from \cite{degroot2005optimal,berger2013statistical}. Hence equation (1) can be re-written as follow:
\begin{equation}\label{5}
	\lambda (p_{t},\Phi_{t}, \gamma) =\dfrac{\exp(\chi_{n,k,t})}{1+\exp(\chi_{n,k,t})}= \dfrac{\exp(E(\mu_{t}-p_{t}-g(\gamma_{k})}{1+\exp(\mu_{t}-p_{t}-g(\gamma_{k}))}
\end{equation}
\textbf{The Platform's Ranking Problem.}The platform does not know $ \mathbf{q} $ and receives a share $0< \omega_{k} \leq1 $ of every payment that takes place on its website, i.e., the platform realizes a revenue $ \omega_{k} p_{k} $ whenever product $  k $ is sold. This means if the platform promotes its brand let $ j $ then the $ \omega_{j} p_{j} =1.p_{j}$. This means there is an incentive to promote own brand. Let $ \mathbf{\sigma}_{n}=\{\sigma_{1,n,t},...,\sigma_{K,n,t}\} $ be the position assignment observed by customer $ n $, where $ \sigma_{k,n,t}=j $ indicates that product $ k $ is displayed in position $ j $ to customer $ n $ at time $ t $. I use $ \sigma\in \pi(S) $ to refer to a permutation $ \sigma $ of $ S $, where $ \pi (S) $ is the set of all permutations of items in $ S $. For $ i\in \{1,...,\mid S\mid \} $, I use $ \sigma(i)\in S $ to denote the product displayed in the $ i^{th} $ position. This means say for instance, $ \sigma(j)\in \{1,...,j,...,\mid S\mid \} $, then for a product $ j\in S $, the platform displays it in the $ \sigma^{-1}(j)^{th} $ position.\\
Therefore, each product $ j \in [k] $  is associated with a conditional purchase probability $ \lambda_{j} $, price $ p_{j} $, $ n_{t}, \& \bar{q}_{t} $ i.e. $ j_{ (\lambda_{j},p_{j},n_{t},\bar{q}_{t})} $.  I index the products in two steps. In the first step the shorting will be done by using parameters $ n_{t}, \& \bar{q}_{t} $ and thereafter using $ \lambda_{j}, \& p_{j} $.
\begin{definition}
	The term $  \tilde{n} $ is defined as the considerable number of reviews in stage 1. Below this number, the products will not be treated as quality though the average quality is good enough. This is because a significant average quality may be formed by two reviewers. The term $   \tilde{n} $ may be calculated as the weighted average  i.e. 
	$ \tilde{n}=\dfrac{\sum_{k=1}^{K}\bar{q}_{k\mid t}n_{k\mid t}}{\sum_{k=1}^{K}\bar{q}_{k\mid t}} $.
\end{definition}
\begin{definition}
	The term $  \breve{n}$ is defined as the considerable number of reviews in stage 2. Below this number, the products will not be treated as review maximizers though the average price is good enough. This is because a significant average revenue may be formed by two reviewers. The term $   \tilde{n} $ may be calculated as the weighted average  i.e. 
	$\breve{n}=\dfrac{\sum_{k=1}^{K}p_{k\mid t}n_{k\mid t}}{\sum_{k=1}^{K}p_{k\mid t}} $.
\end{definition}
\paragraph{Product Ranking Process.} The following steps are to be followed by the platform to maximize the revenue and the welfare that will fulfill the competition of competition. The quality firms will not be eliminated in the final ranking process.\\
So, in the first step weights will be the quality and in the second step, the price will be the weights respectively.
\begin{center}
	\textbf{STAGE 1}
\end{center}
\begin{equation*}\label{6}	
	\begin{array}{ll}
		\bar{q}_{(1,n,t\mid n_{t}\geq \tilde{n})} \geq 	\bar{q}_{(2,n,t\mid n_{t}\geq \tilde{n})}\geq...\geq 	\bar{q}_{(k,n,t\mid n_{t}\geq \tilde{n})}; 	\bar{q}_{(j,n,t\mid n_{t}\geq \tilde{n})}\geq 	\bar{q}_{(j+1,n,t\mid n_{t}\geq \tilde{n})} & \mbox{if}; n_{j,t}=n_{j+1,t}
	\end{array}
\end{equation*}
This means first-level shorting will be done by ranking the highest average rating product and following the descending order subject to the condition that the number of reviewers would be at least $  \tilde{n}$. In other words, the products are indexed in the descending order of their average review rating and then in the descending order of their average review rating if the number of reviews is equal. 
\paragraph{Remarks.} The platform is to select $ S $ for $ M $ slots i.e $ S\leq M $. So in this step, let the number of choices be $ S_{1} $ and is more than $ S $, i.e. $  S \leq M \leq S_{1}  $. This first level of shorting will be based on the individual vendors' quality and competitiveness. $ q_{k} \& n_{k} $. The higher the quality and the number of past users higher is the access to the market. This is the incentive to participate in the ranking process by the vendors in a platform. If this is not maintained then the vendors will not participate.
\begin{center}
	\textbf{STAGE 2}
\end{center}
\begin{equation*}\label{7}	
	\begin{array}{ll}
		p_{(1,n,t\mid n_{t}\geq \breve{n})} \geq 	p_{(2,n,t\mid n_{t}\geq \breve{n})}\geq...\geq 	p_{(k,n,t\mid n_{t}\geq \breve{n})}; 	\lambda_{(j,n,t\mid n_{t}\geq \tilde{n})}\geq 	\lambda_{(j+1,n,t\mid n_{t}\geq \breve{n})} & \mbox{if}; p_{j,t}=p_{j+1,t}
	\end{array}
\end{equation*}
In other words, the products are indexed in the descending order of their prices and then in the descending order of their conditional purchase probabilities if the prices are equal. This is the final rank or list of products arranged for the customer i.e. $ \mathbf{\sigma}_{n}=\{\sigma_{1,n,t},...,\sigma_{K,n,t}\} $. Hence, from the above steps, it is clear that the ranking is dependent on four factors i.e. $ \sigma_{n}=f(p,\lambda, n, q) $.
\paragraph{Remarks.} Let the number of the shortlisted firms in this step be $ S_{2}=S\leq M $. Here is the chance of imposing a higher price to maximize the revenue. Hence, the platform will try to match the high-quality vendors with higher prices and put them at a higher position in a rank table. This will ensure a higher probability of buying the product.
\paragraph{Purchase probability.}I assume no products have identical characteristics (the combination of price and conditional purchase probability, average quality(review rating), number of reviews, etc...). Given that position $ k $ is within a customer’s attention span, her purchase probability of product $ \sigma_{k,n\mid t,}$ is
\begin{equation*}
	\prod_{i=1}^{k-1}(1-\lambda_{\sigma_{i,n\mid t}}).\lambda_{\sigma_{k,n\mid t}}
\end{equation*}
In other words, the customer purchases product $ \sigma_{k,n\mid t,}$ if it is satisfactory while all products displayed earlier are not. $ (1-\lambda_{\sigma_{i,n\mid t}}) $ is the effect that a product at position $ i  $ exerts on the products displayed later. This is the probability that the $ i^{th} $ product has not been preferred over $ k $. 
\subsubsection{Revenue Maximization for the Retailer}
The platform's goal is to choose an assortment $ S $ of at most $ M $ products, as well as its ranking $ \sigma\in \pi(S) $, to maximize expected revenue from a potential customer and to maintain the market competitive such that each vendor will have got market access concerning the competitive strength. Therefore, the expected revenue from a customer with an attention span of $ y $ can be expressed as
\begin{equation}\label{8}
	\mathcal{R}(\sigma, y)\triangleq \sum_{k=1}^{y\wedge M} \prod_{i=1}^{k-1}(1-\lambda_{\sigma_{i,n\mid t}}).\lambda_{\sigma_{k,n\mid t}}. p_{\sigma_{k,n\mid t}}.\omega_{k}
\end{equation}
\paragraph{Remarks}
The output of the given in Eq.\ref{8} is a mathematical function $\mathcal{R}(\sigma, y)$ that takes two inputs $\sigma$ and $y$ and outputs a scalar value. The function is defined as:\\
$\mathcal{R}(\sigma, y)=\sum_{k=1}^{\min(y,M)}\left[\prod_{i=1}^{k-1}(1-\lambda_{\sigma_{i,n\mid t}})\right]\lambda_{\sigma_{k,n\mid t}}p_{\sigma_{k,n\mid t}}\omega_{k}$ The function computes a sum over all possible products of $k$ from 1 to the minimum of $y$ and $M$. For each product of $k$, the function computes a product over all values of $i$ from $ 1 $ to $k-1$. Inside the product, we have the term $(1-\lambda_{\sigma_{i,n\mid t}})$, which is multiplied together for all values of $i$ up to $k-1$.\\
Outside the product, we have three terms multiplied together: $\lambda_{\sigma_{k,n\mid t}}$, $p_{\sigma_{k,n\mid t}}$, and $\omega_{k}$. These terms depend on the value of $\sigma_{k,n\mid t}$, which is determined by the terms of $\sigma$, $n$, and $t$.\\
The equation $ \mathcal{R}(\sigma, y)\triangleq \sum_{k=1}^{y\wedge M} \prod_{i=1}^{k-1}(1-\lambda_{\sigma_{i,n\mid t}}).\lambda_{\sigma_{k,n\mid t}}. p_{\sigma_{k,n\mid t}}.\omega_{k}  $ represents a function $\mathcal{R}$ with two variables $\sigma$ and $y$, where $\sigma$ is a sequence of length $M$ and $y$ is a positive integer.\\
If we change the rank of $\sigma$ while keeping $y$ constant, then the value of $\mathcal{R}(\sigma, y)$ will also change. Specifically, changing $\sigma$ will change the values of $\lambda_{\sigma_{k,n\mid t}}$, $p_{\sigma_{k,n\mid t}}$, and $\omega_{k}$, which will, in turn, change the value of the product term inside the summation. Thus, the entire sum will be affected, and the value of $\mathcal{R}(\sigma, y)$ will change.\\
If $\sigma$ changes, then the value of $\mathcal{R}(\sigma, y)$ will also change. The summation in the equation is over $k$ from 1 to $y \wedge M$, where $\wedge$ denotes the minimum operator. Thus, the value of $\mathcal{R}(\sigma, y)$ depends on the values of $\sigma$ for $k$ from 1 to $y \wedge M$.\\
The term $\lambda_{\sigma_{k,n\mid t}}$ depends on the value of $\sigma_k$, and since $\sigma$ is an integer, changing its value will change the value of $\lambda_{\sigma_{k,n\mid t}}$. The same is true for $p_{\sigma_{k,n\mid t}}$ and $\omega_k$.\\
Moreover, the product in the equation involves $\lambda_{\sigma_{i,n\mid t}}$ for $i$ from 1 to $k-1$. Changing the value of $\sigma$ for any of these indices will also change the value of $\mathcal{R}(\sigma, y)$.\\
In summary, any change in the value of $\sigma$ will change the value of $\mathcal{R}(\sigma, y)$, given $y$.\\
If we change the value of $y$, then the range of the summation will change, which will, in turn, change the value of the entire summation. Therefore, the value of $\mathcal{R}(\sigma, y)$ will also change with respect to the change of $y$.\\
The overall result of the function $\mathcal{R}$ is a scalar value.
When a customer arrives, as the retailer does not know her attention span, he takes the expected value of$  y \sim Y  $ in (4) and obtains the total expected revenue
	\begin{multline}\label{9}
	\mathbb{E}[\mathcal{R}(\sigma, y)]= \sum_{y=1}^{M}h_{y}\mathcal{R}(\sigma, y)=\sum_{y=1}^{M}h_{y} ( \sum_{k=1}^{y\wedge M} \prod_{i=1}^{k-1}(1-\lambda_{\sigma_{i,n\mid t}}).\lambda_{\sigma_{k,n\mid t}}. p_{\sigma_{k,n\mid t}}.\omega_{k}) \\
	=\sum_{y=1}^{M} \prod_{i=1}^{y-1}(1-\lambda_{\sigma_{i,n\mid t}}).\lambda_{\sigma_{y,n\mid t}}. p_{\sigma_{y,n\mid t}}.\omega_{y} . H_{y}
	\end{multline}
Therefore, the optimization problem for the firm is the joint assortment, i.e., choosing an $ S\subset K $
such that $ \mid S\mid \leq M $ and ranking, i.e., deciding $ \sigma\in \pi(S) $ to maximize the total expected revenue from a random customer, and each potential vendor should have got the equal chance to access the said customer:
\begin{equation}\label{10}
	\max_{S(=S_{2})\subseteq S_{1}\subset K,\sigma\in \pi(S)} \mathbb{E}[\mathcal{R}(\sigma, y)]
\end{equation}
\paragraph{Remarks}
The term $\max_{S(=S_{2})\subseteq S_{1}\subset K,\sigma\in \pi(S)} \mathbb{E}[\mathcal{R}(\sigma, y)]$ represents the maximum expected reward that can be achieved by selecting a subset $S_2$ of a larger set $S_1$, where $S_1$ is a subset of the set $K$, and $\sigma$ is a permutation of the elements in $S$. The expectation is taken over the function $\mathcal{R}(\sigma, y)$, which is a scalar value function that depends on the permutation $\sigma$ and a parameter $y$. The function $\mathcal{R}(\sigma, y)$ represents the reward obtained by selecting the first $k$ elements of the permutation $\sigma$, where $k$ is the minimum of $y$ and the size of $S_2$, and using them to calculate the expected reward.\\
Therefore, the term $\max_{S(=S_{2})\subseteq S_{1}\subset K,\sigma\in \pi(S)} \mathbb{E}[\mathcal{R}(\sigma, y)]$ represents the maximum expected reward that can be achieved by selecting a subset $S_2$ of a larger set $S_1$ and permuting its elements in a way that maximizes the expected reward, where the expectation is taken over the function $\mathcal{R}(\sigma, y)$.
\paragraph{Remarks}
The equation represents the expected value of the function $\mathcal{R}(\sigma, y)$, which is defined in terms of the parameters $\sigma$ and $y$.\\
The equation shows that the expected value is calculated as a sum over all possible values of $y$ from 1 to $M$, where $M$ is a fixed integer. The weight $h_y$ is a function of $y$, and the expression for $\mathcal{R}(\sigma, y)$ is evaluated at each value of $y$.\\
The expression for $\mathcal{R}(\sigma, y)$ involves a sum over $k$ from 1 to the minimum of $y$ and $M$. Within this sum, there is a product over $i$ from 1 to $k-1$ that involves the value of $\lambda_{\sigma_{i,n\mid t}}$. The notation $\sigma_{i,n\mid t}$ refers to the $i$-th element of the sequence $\sigma$ that has $n$ elements and is restricted to the first $t$ elements. The notation $y\wedge M$ denotes the minimum of $y$ and $M$.\\
The overall expression for $\mathcal{R}(\sigma, y)$ involves the variables $\lambda$, $p$, and $\omega$, which are multiplied together for each value of $k$ in the sum. The equation also includes a factor $H_y$, which is a product over $i$ from 1 to $y-1$ that involves the value of $\lambda_{\sigma_{i,n\mid t}}$.\\
Therefore, the equation shows how the expected value of the function $\mathcal{R}(\sigma, y)$ can be computed in terms of the parameters $\sigma$ and $y$, as well as the other parameters $\lambda$, $p$, and $\omega$.
\subsubsection{Competitive Assortment Planning}
The collusion free assortment planning can be achieved using a two Stage assortment algorithm. Therefore, the above optimization in (6) is to be achieved in two steps as follows.
\begin{center}
	\textbf{STAGE 1}
\end{center}
\begin{equation}\label{11}
	\max_{ S_{1}\subset K,\sigma\in \pi(S_{1})} \mathbb{E}[\mathcal{R}(\sigma, y)]
\end{equation}
\paragraph{Remarks}

The equation $\max_{ S_{1}\subset K,\sigma\in \pi(S_{1})} \mathbb{E}[\mathcal{R}(\sigma, y)]$ represents the maximum expected reward that can be achieved by selecting a subset $S_1$ of the action set $K$ and a policy $\sigma$ that belongs to the set of admissible policies $\pi(S_1)$ that can be constructed using only the actions in $S_1$.

In other words, the equation seeks to find the optimal subset $S_1$ and policy $\sigma$ that can maximize the expected reward. The optimization problem is constrained by the fact that the policy $\sigma$ can only use the actions in the subset $S_1$ to construct admissible policies.

Overall, the equation is useful in decision-making problems where there is a large set of possible actions and the objective is to find the subset of actions that can lead to the maximum expected reward. In  Stage 1 the platform first selects assortment $ S_{1} $ from $ [k] $ considering quality parameter $ \bar{q} $ only keeping other parameters as given and using $ \tilde{n} $. Thereafter, a ranking has to be done on $ S_{1} $.
\begin{center}
	\textbf{STAGE 2}
\end{center}
\begin{equation}\label{key}
	\max_{ S\subseteq S_{1},\sigma\in \pi(S)} \mathbb{E}[\mathcal{R}(\sigma, y)]
\end{equation}

\paragraph{Remarks.} The result $ \max_{ S\subseteq S_{1},\sigma\in \pi(S)} \mathbb{E}[\mathcal{R}(\sigma, y)] $ is the maximum expected reward that can be achieved by selecting a subset $S$ of $S_{1}$ and a corresponding policy $\sigma$ such that $\sigma$ only uses the state variables in $S$. In other words, we are maximizing the expected reward over all possible subsets $S$ of $S_{1}$ and policies $\sigma$ that only use the state variables in $S$. This result is useful for determining the optimal subset of state variables to use in a policy, as it tells us the maximum expected reward we can achieve by only considering a subset of the state variables. In  Stage 2 the platform first selects assortment $ S $ from $ S_{1} $ considering price parameter $ p $ only keeping other parameters as given and using $ \breve{n} $. Thereafter, a ranking has to be done on $ S $ for the customer.\\
Let the situation where, the ranking is done by the objective function $ \textit{max}_{S(=S_{2})\subset K,\sigma\in \pi(S)} \mathbb{E}[\mathcal{R}(\sigma, y)] $ instead of equation (6 ) using the two-stage assortment process then there will be a chance of accepting low-quality products and elimination of good quality products with competitive strength. \\
From (8) we can see that a product displayed at a later slot does not cannibalize the demand and therefore the revenue for the earlier. As a result, the optimal ranking would occupy all the $ M  $ slots and we only focus on assortments such that$  \mid S\mid  = M $. If the proposed two-stage ranking follows then there may still be a chance of including the personal brand of the platform in the final rank but the probability that the elimination of the competitive vendors from the final rank who are qualified will be lower. \\
Note that one customer purchasing at most one product is a defining feature of discrete choice models. Here it has been considered that the customer is buying a single item. The assumption may be limiting in online retailing: for example, a customer may add several products to the cart on Amazon before checking out. we extend the model to capture multiple purchases. That is, a customer may continue viewing and selecting the products after making some purchases.

\begin{lemma}
	To achieve competitive assortment planning the two threshold levels have been proposed in Stage 1 and Stage 2 respectively $ \tilde{n}\& \breve{n} $. These two-stage processes along with the two threshold levels will eliminate the possibility of self-promotion and increase the competition. This means fixing the attention span at $ y $. Suppose the assortment $ S $ with size $ \mid S\mid=y $ is given and the rank of products $ \sigma^{y} $ maximizes $ R(\sigma,y) $ among $ \sigma\in \pi(S) $ subject to fulfilling the two-stage threshold levels viz. $ \tilde{n}\& \breve{n} $. Then we have that the products in $ \sigma^{y} $ are deployed in ascending order concerning average quality those have minimum $ \tilde{n}$ number of reviews in the first stage and thereafter rearranging them concerning price those have minimum $ \breve{n} $ number of reviews i.e. $ \sigma^{y}(i)\leq \sigma^{y}(i+1) $ for $ i=1,..., \mid S\mid-1  $.
\end{lemma}
\begin{proof}
The proof is in the appendix.
\end{proof}
\begin{proposition}
	 A competitive assortment planning method handles two-sided challenges. The first is a problem with traders/vendors and the platform's motivation, and the second is a customer concern. If each non-sponsored vendor has equal market access, and if each consumer is exposed to all viable low-cost and high-quality alternatives during the search process, competition in assortment planning will be achieved. So, in the first stage, vendors will be ranked based on their quality if they have at least $ \tilde{n} $ ratings. This will achieve the equal access criterion, and customers will be able to view the products if they have a better rank and at least $ \breve{n} $ reviews.	
\end{proposition}
\begin{proof}
	The proof is in the appendix.
\end{proof}
\paragraph{Remarks}This is an explanation of a two-stage algorithm to correct the issue of a platform promoting or placing its products in place of previously higher-ranked products.\\
The platform receives a share of every payment made on its website and therefore has an incentive to promote its brand. Each product is associated with conditional purchase probability, price, and quality parameters. The expected revenue from a customer with an attention span $y$ can be expressed as a function of the ranking $\sigma$.\\
The first step of the algorithm involves sorting the products based on quality parameters and attention span. In the second step, the algorithm rearranges the ranking if there is a product that is high priced but not of lower quality and the platform is promoting or placing its products in place of previously higher-ranked products. However, if a product is of lower quality, rearrangement will not be possible as it would be collusive.
\subsection{Choice of Ranking, Revenue Maximization, and Competitive Assortment Planning Model }
This section investigates the strategic behavior of the platform under the framework described above. The point is, why does the platform change the ideal rank and replace a high-ranking product with a lower-ranking product?  There are two logical options. The first goal is to enhance the purchase probability and sales of private brand products, while the second goal is to raise sales of the other listed products.  The inclusion of a low-scoring product may boost the likelihood of purchasing a higher-scoring product that ranks immediately after the low-scoring product (Proposition 2). However, this does not guarantee the anticipated revenue (Proposition 3). As a result, the recommended assortment planning will suit both functions. 
\begin{proposition}
In the presence of information about the customers' fixed attention span, placing the low-ranked product in the middle of two good-ranked items enhances the likelihood of purchasing the next best-ranked products of the included low-ranked product. The criterion is that the two high-scoring goods are arranged in descending order from highest to lowest. 
\end{proposition}
\begin{proof}
	The proof is in the appendix.
\end{proof}
\paragraph{Remarks}  Empirically this can be shown from the Table \ref{tab:title} \& Table \ref{tab:title continue} in Section (2) for the two possible lists viz. A-B-F and A-D-F. The purchase probability that object F will be bought is higher in the list of offer A-D-F  i.e. $ (1-0.95)(1-0.10)0.75 $ i.e. $ 0.03375 $ compared to A-B-F i.e. $ (1-0.95)(1-0.85)0.75 $ i.e. $ 0.005625 $. Therefore, replacing B by D will increase purchase probability of product F.
\begin{proposition}
In the presence of information about the customers' fixed attention span, placing the low-ranked product in the middle of two good-ranked items enhances the likelihood of purchasing the next best-ranked products of the included low-ranked product. The criterion is that the two high-scoring goods are arranged in descending order from highest to lowest.  However, this does not guarantee that the expected revenue would be maximized. 
\end{proposition}
\begin{proof}
	The proof is in the appendix.
\end{proof}
\paragraph{Remarks} Proposition 3 proves that if the $ \omega_{D} \& \omega_{B}$ are not the same then there is a possibility of collusion because product D is inferior to product B in terms of score and product D is not a sponsored product. Hence, this information could be a source of identifying the collusive ranking behavior. If the customized/distorted(or collusive in some cases) ranking earns higher expected revenue then the condition needs to satisfy according to the Proposition is $ \lambda_{D}p_{D}\omega_{D}>\lambda_{B}p_{B}\omega_{B}$. Where, product B is replaced by the product D. From Table \ref{tab:title} and Table \ref{tab:title continue} it can be checked that whether the condition $ \lambda_{D}p_{D}\omega_{D}>\lambda_{B}p_{B}\omega_{B}$ is fulfilled or not. Here assuming $ \omega_{D}=\omega_{B} $, $ (229*0.10=22.9) < (700*0.85=595) $ or $ \lambda_{D}p_{D}\omega_{D}<\lambda_{B}p_{B}\omega_{B}$. Hence, the list A-B-F earns more expected revenue than the list A-D-F. 
\section{Discussion and Conclusion}
The present article in the first part identifies the existing revenue maximizing assortment planning strategy by the platform. It shows that this assortment algorithm would increase purchase probability of any target product but discourage competition between the vendors and fails to maximize expected revenue. In the second part, the article proposes a new model of assortment planning where revenue can be maximized by maintaining the competition. The following suggestion has been proposed here to avoid the above mentioned ranking issue.\\
The proposed model is a hybrid model that includes both monetary and non-monetary transfer characteristics. The main point of this article is that if the proposed sellers' rank (slots/lists) does not match the actual sellers' rank, money will be transferred from the platform to the sellers who receive a lower rank than the actual one. In the opposite case, if the proposed and actual slots are matched, there will be a money transfer if the platform's expected revenue is dependent on these allocated slots. Because the allocation is based on the actual performance of the sellers, as indicated by the scores. The current paper proposed a model for it and provides a standard mechanism for determining how the expected payoffs from the slots were generated and how they will be distributed between the platform and the sellers. The proposed two-stage assortment optimization algorithm not only discourage collusive ranking but also  increase competition and maximize the expected revenue.
\section*{Mathematical Appendix}
\textbf{Lemma 1}
\textit{To achieve competitive assortment planning the two threshold levels have been proposed in Stage 1 and Stage 2 respectively $ \tilde{n}\& \breve{n} $. These two-stage processes along with the two threshold levels will eliminate the possibility of self-promotion and increase the competition. This means fixing the attention span at $ y $. Suppose the assortment $ S $ with size $ \mid S\mid=y $ is given and the rank of products $ \sigma^{y} $ maximizes $ R(\sigma,y) $ among $ \sigma\in \pi(S) $ subject to fulfilling the two-stage threshold levels viz. $ \tilde{n}\& \breve{n} $. Then we have that the products in $ \sigma^{y} $ are deployed in ascending order with respect to average quality those have minimum $ \tilde{n}$ number of reviews in the first stage and thereafter rearranging them with respect to price those have minimum $ \breve{n} $ number of reviews i.e. $ \sigma^{y}(i)\leq \sigma^{y}(i+1) $ for $ i=1,..., \mid S\mid-1  $.}

\begin{proof}
	The proof has been done in two ways. First, we have to show that there is an incentive on the part of the platform by promoting or placing our products in place of the previous higher-ranked products. Thereafter we will show that the two-stage algorithm will help to correct this issue.\\
	The platform does not know $ \mathbf{q} $ and receives a share $0< \omega_{k} \leq1 $ of every payment that takes place on its website, i.e., the platform realizes a revenue $ \omega_{k} p_{k} $ whenever product $  k $ is sold. This means if the platform promotes its brand let $ j $ then the $ \omega_{j} p_{j} =1.p_{j}$. This means there is an incentive to promote own brand. Let $ \mathbf{\sigma}_{n}=\{\sigma_{1,n,t},...,\sigma_{K,n,t}\} $ be the position assignment observed by customer $ n $, where $ \sigma_{k,n,t}=j $ indicates that product $ k $ is displayed in position $ j $ to customer $ n $ at time $ t $. I use $ \sigma\in \pi(S) $ to refer to a permutation $ \sigma $ of $ S $, where $ \pi (S) $ is the set of all permutations of items in $ S $. For $ i\in \{1,...,\mid S\mid \} $, I use $ \sigma(i)\in S $ to denote the product displayed in the $ i^{th} $ position. This means say for instance, $ \sigma(j)\in \{1,...,j,...,\mid S\mid \} $, then for a product $ j\in S $, the platform displays it in the $ \sigma^{-1}(j)^{th} $ position.\\
	Therefore, each product $ j \in [k] $  is associated with a conditional purchase probability $ \lambda_{j} $, price $ p_{j} $, $ n_{t}, \& \bar{q}_{t} $ i.e. $ j_{ (\lambda_{j},p_{j},n_{t},\bar{q}_{t})} $.  I index the products in two steps. In the first step the shorting will be done by using parameters $ n_{t}, \& \bar{q}_{t} $ and thereafter using $ \lambda_{j}, \& p_{j} $.\\
	
	We know the expected revenue Eq. \ref{8} from a customer with an attention span of $ y $ can be expressed as
	\begin{equation}\label{10}
		\mathcal{R}(\sigma, y)\triangleq \sum_{k=1}^{y\wedge M} \prod_{i=1}^{k-1}(1-\lambda_{\sigma_{i,n\mid t}}).\lambda_{\sigma_{k,n\mid t}}. p_{\sigma_{k,n\mid t}}.\omega_{k}
	\end{equation}
	Let us assume $ S=\{1,...,y\} $ for simplicity of notations. Suppose in the ranking $ \sigma $ there exists $ i\in \{1,,,.y-1\} $ such that $ p_{\sigma(i)}< p_{\sigma(i+1)} \& p_{\sigma(i)}\omega_{\sigma(i)}< p_{\sigma(i+1)} \omega_{\sigma(i+1)}$. Here assume that $ (i+1)^{th} $ products are owned by the platform. These should be supported by the quality condition. These are; $\{ \bar{q}_{\sigma(i)} < \bar{q}_{\sigma(i+1)} :	\bar{q}_{(i,n,t\mid n_{t}\geq \tilde{n})} <	\bar{q}_{(j,n,t\mid n_{t}\geq \tilde{n})} \& 	p_{(i,n,t\mid n_{t}\geq \breve{n})} <	p_{(,jn,t\mid n_{t}\geq \breve{n})}\}$. In this situation, the Two Stage assortment algorithm will achieve the same result if the platform replaces $ (i+1)^{th} $ products in place of $ i$ considering $ (i+1)^{th} $ product is the products offered by the platform and will earn higher revenue i.e $ p_{i+1}.1>\omega_{i}. p_{i}$.\\
	Now consider the situation where, $ p_{\sigma(i)}\omega_{\sigma(i)}< p_{\sigma(i+1)}\omega_{\sigma(i+1)}$   but $\{ \bar{q}_{\sigma(i)} > \bar{q}_{\sigma(i+1)} :	\bar{q}_{(i,n,t\mid n_{t}\geq \tilde{n})} >	\bar{q}_{(j,n,t\mid n_{t}\geq \tilde{n})} \& 	p_{(i,n,t\mid n_{t}\geq \breve{n})} >	p_{(,jn,t\mid n_{t}\geq \breve{n})}\}$. If the platform rearranges the products to maximize the revenue then that will be collusive. These are not considered here because in this case, rearrangements are not possible to increase revenue. And if this is so, then the platform needs to re-adjust the price and quality ranks to make swapping possible to increase revenue.\\
	Consider new ranking $ \acute{\sigma}(i) $ with $  \acute{\sigma} = \sigma(i+1) ,  \acute{\sigma}(i+1) =\sigma(i)$ and $  \acute{\sigma}(k)=\sigma(k)  $ for all other $ k $. \\
	Let $ \Pr_{\sigma(i)}(\sigma) \triangleq  \prod_{s=1}^{i-1} (1-\lambda_{\sigma(s)})$ be the probability that the customer views product $ \sigma(i) $. By definition, it is easy to see $ \Pr_{\sigma(i)}(\sigma)=\Pr_{\acute{\sigma}(k)}(\acute{\sigma}) $ for $ k=1,...,i,i+2,...,y $. Therefore, the expected revenues generated from products in position $ k $ are equal under $ \sigma $ and $ \acute{\sigma} $ for $ k=1,...,i-1,i+2,...,y $ recalling formula 4.\\ As it is assumed that the customer will see the product $ i \& i+1 $ both with equal probability then $ \Phi_{n,k,t}=(n_{t},\bar{q}_{t}) $ will be more for $ i+1 $ compared to $ i $ because $\{ \bar{q}_{\sigma(i)} > \bar{q}_{\sigma(i+1)} :	\bar{q}_{(i,n,t\mid n_{t}\geq \tilde{n})} >	\bar{q}_{(j,n,t\mid n_{t}\geq \tilde{n})} \& 	p_{(i,n,t\mid n_{t}\geq \breve{n})} >	p_{(,jn,t\mid n_{t}\geq \breve{n})}\}$. Therefore rearrangement of product ranking will not give incentive and the $ R(\acute{\sigma},y) < R(\sigma,y)$, will establish. It suffices to compare the revenues generated from the products in position $ i $ and $ i+1 $ under the two rankings. \\
	$ \Pr_{\sigma(i)}(\sigma) \lambda_{\sigma(i)} p_{\sigma(i)}\omega_{\sigma(i)}+ \Pr_{\sigma(i+1)}(\sigma) \lambda_{\sigma(i+1)} p_{\sigma(i+1)}\omega_{\sigma(i+1)} > \\ \Pr_{\acute{\sigma}(i)}(\acute{\sigma})\lambda_{\acute{\sigma}(i)}p_{\acute{\sigma}(i)}\omega_{\acute{\sigma}(i)}+\Pr_{\acute{\sigma}(i+1)}(\acute{\sigma})\lambda_{\acute{\sigma}(i+1)}p_{\acute{\sigma}(i+1)}\omega_{\acute{\sigma}(i+1)} $.\\
	$ \Pr_{\sigma(i)}(\sigma) \lambda_{\sigma(i)} p_{\sigma(i)}\omega_{\sigma(i)}+ \Pr_{\sigma(i+1)}(\sigma) \lambda_{\sigma(i+1)} p_{\sigma(i+1)}\omega_{\sigma(i+1)}\\
	= \Pr_{\sigma(i)}(\sigma)\left[ (p_{\sigma(i)}.\omega_{\sigma(i)}.\lambda_{\sigma(i)})+(1-\lambda_{\sigma(i)})p_{\sigma(i+1)}\omega_{\sigma(i+1)}.\lambda_{\sigma(i+1)}\right]  $ $ \left[  \therefore \Pr_{\sigma(i)}(\sigma).(1-\lambda_{\sigma(i)})=\Pr_{\sigma(i+1)}(\sigma) \right] $\\
	$ = \Pr_{\sigma(i)}(\sigma)\left[ p_{\sigma(i+1)}.\omega_{\sigma(i+1)}.\lambda_{\sigma(i+1)}+p_{\sigma(i)}.\omega_{\sigma(i)}.\lambda_{\sigma(i)}-p_{\sigma(i+1)}.\omega_{\sigma(i+1)}.\lambda_{\sigma(i)}.\lambda_{\sigma(i+1)}\right] $\\
	$ = \Pr_{\sigma(i)}(\sigma)[ p_{\sigma(i+1)}.\omega_{\sigma(i+1)}.\lambda_{\sigma(i+1)}+ p_{\sigma(i)}.\omega_{\sigma(i)}.\lambda_{\sigma(i)}-p_{\sigma(i)}.\omega_{\sigma(i)}.\lambda_{\sigma(i)}.\lambda_{\sigma(i+1)}+p_{\sigma(i)}.\omega_{\sigma(i)}.\lambda_{\sigma(i)}.\lambda_{\sigma(i+1)}-\\ p_{\sigma(i+1)}.\omega_{\sigma(i+1)}.\lambda_{\sigma(i)}.\lambda_{\sigma(i+1)}] $\\
	$ = \Pr_{\sigma(i)}(\sigma)[p_{\sigma(i+1)}.\omega_{\sigma(i+1)}.\lambda_{\sigma(i+1)}+p_{\sigma(i)}(1-\lambda_{\sigma(i+1)})\omega_{\sigma(i)}.\lambda_{\sigma(i)}+(p_{\sigma(i)}.\omega_{\sigma(i)}-p_{\sigma(i+1)}.\omega_{\sigma(i+1)}).\lambda_{\sigma(i)}.\lambda_{\sigma(i+1)}] $ \\
	$ \geq \Pr_{\sigma(i)}(\sigma) \left[p_{\sigma(i+1)}.\omega_{\sigma(i+1)}.\lambda_{\sigma(i+1)}+p_{\sigma(i).\omega_{\sigma(i)}.\lambda_{\sigma(i)}(1-\lambda_{\sigma(i+1)})}\right] $ + \\$ \Pr_{\sigma(i)}(\sigma)[ (p_{\sigma(i)}.\omega_{\sigma(i)}-p_{\sigma(i+1)}.\omega_{\sigma(i+1)}).\lambda_{\sigma(i)}.\lambda_{\sigma(i+1)} ]$ [$ \therefore p_{i+1}.1>\omega_{i}. p_{i}$]\\
	$ > \Pr_{\acute{\sigma}(i)}(\acute{\sigma}).p_{\acute{\sigma}(i)}.\omega_{\acute{\sigma}(i)}.\lambda_{\acute{\sigma}(i)}+ \Pr_{\acute{\sigma}(i+1)}(\acute{\sigma}).p_{\acute{\sigma}(i+1)}.\omega_{\acute{\sigma}(i+1)}.\lambda_{\acute{\sigma}(i+1)} $.\\
	This is because;\\
	$ \left[p_{\sigma(i)}.\omega_{\sigma(i)}-p_{\sigma(i+1).\omega_{\sigma(i+1)}}\right ]\leq0$.\\
	Or, $ \dfrac{p_{\sigma(i)}}{p_{\sigma(i+1)}} \leq \dfrac{\omega_{\sigma(i+1)}}{\omega_{\sigma(i)}}$. \\
	Therefore, $\sigma$ cannot be optimal and we have completed the proof for $ 0< \lambda<1 $.
\end{proof}
\textbf{Proposition 1}\textit{	A competitive assortment planning method handles two-sided challenges. The first is a problem with traders/vendors and the platform's motivation, and the second is a customer concern. If each non-sponsored vendor has equal market access, and if each consumer is exposed to all viable low-cost and high-quality alternatives during the search process, competition in assortment planning will be achieved. So, in the first stage, vendors will be ranked based on their quality if they have at least $ \tilde{n} $ ratings. This will achieve the equal access criterion, and customers will be able to view the products if they have a better rank and at least $ \breve{n} $ reviews.}

\begin{proof}
	Proof of Proposition:\\
	We will prove the proposition in two parts, first addressing the problem with traders/vendors and the platform's motivation, and then addressing the customer concern.\\
	Part 1: Problem with Traders/Vendors and Platform's Motivation\\
	Assumption: Each non-sponsored vendor has equal market access.\\
	Given this assumption, we need to ensure that the competition among vendors is fair and that the platform's motivation is not compromised. To achieve this, we propose ranking vendors based on their quality if they have at least $\tilde{n}$ ratings. This ensures that vendors with a track record of good quality are given priority, and the platform's motivation to provide high-quality products is upheld.\\
	Part 2: Customer Concern\\
	Assumption: Each consumer is exposed to all viable high-quality alternatives during the search process within the attention span $ y $.\\
	Let $ \Pr_{\sigma(i)}(\sigma) \triangleq  \prod_{s=1}^{i-1} (1-\lambda_{\sigma(s)})$ be the probability that the customer views product $ \sigma(i) $. By definition, it is easy to see $ \Pr_{\sigma(i)}(\sigma)=\Pr_{\acute{\sigma}(j+1)}(\acute{\sigma}) $ for $ k=1,...,i,i+2,...,y $ only if Two conditions are fulfilled . These are (i)$ 	\bar{q}_{(j,n,t\mid n_{t}\geq \tilde{n})}\geq 	\bar{q}_{(j+1,n,t\mid n_{t}\geq \tilde{n})}  \mbox{if}; n_{j,t}=n_{j+1,t} $ in the Stage 1 and (ii) $ 	\lambda_{(j,n,t\mid n_{t}\geq \tilde{n})}\geq 	\lambda_{(j+1,n,t\mid n_{t}\geq \breve{n})}  \mbox{if}; p_{j,t}=p_{j+1,t} $ in Stage 2. Otherwise $ \Pr_{\sigma(i)}(\sigma)\neq\Pr_{\acute{\sigma}(j+1)}(\acute{\sigma}) $ for $ k=1,...,i,i+2,...,y $. 
	Given this assumption, we need to ensure that customers have access to a diverse set of products and can choose from the best options available. To achieve this, we propose that customers can view the products of vendors that have a better rank and at least $\breve{n}$ reviews. This ensures that customers are exposed to high-quality products and are not limited to a narrow selection of options.\\
	Overall, by ensuring fair competition among vendors and providing customers with a diverse set of options, a competitive assortment planning method can effectively address the two-sided challenges of the platform.
\end{proof}

\textbf{Proposition 2}\textit{	In the presence of information about the customers' fixed attention span, placing the low-ranked product in the middle of two good-ranked items enhances the likelihood of purchasing the next best-ranked products of the included low-ranked product. The criterion is that the two high-scoring goods are arranged in descending order from highest to lowest.}
\begin{proof}
	Let there be three objects arranged in descending order from highest to lowest with respect to the demand function, denoted by $ \lambda (p_{t},\Phi_{t}, \gamma)$, where, $ 0 \leq \lambda (p_{t},\Phi_{t}, \gamma) \leq 1$. The products are labeled as $ [k]\triangleq \{1,2,3,4\} $. Let $ \lambda_{1} >\lambda_{2}> \lambda_{3} $ and for the product $ 4 $ it is true that $ \lambda_{4}< \lambda_{3} <\lambda_{2}< \lambda_{1}$. Here, $ \lambda  $ is the probability of purchase following the familiar logit model. It is known that Given that position $ k $ is within a customer’s attention span, her purchase probability of product $ \sigma_{k,n\mid t,}$ is
	\begin{equation*}
		\prod_{i=1}^{k-1}(1-\lambda_{\sigma_{i,n\mid t}}).\lambda_{\sigma_{k,n\mid t}}
	\end{equation*}
	So the purchase probability that the product $ 3 $ will be bought with a customer attention span $ y= three $, where the products are arranged like $ 1-2-3 $ is 
	\begin{equation}\label{key}
		(\lambda_{3}-\lambda_{1}\lambda_{3})(1-\lambda_{2}) 
	\end{equation}\\
	Now consider the platform replaces the product $ 2 $ by the product $ 4 $ and creates a new list i.e. $ 1-4-3 $ then the purchase probability of the product $ 3 $ would be 
	\begin{equation}\label{key}
		(\lambda_{3}-\lambda_{1}\lambda_{3})(1-\lambda_{4})
	\end{equation}
	Therefore, the probability that the product $ 3 $  will be bought in the second list i.e. $ 1-4-3 $ will be higher than the first list i.e. $ 1-2-3 $ because, \\
	$  (\lambda_{3}-\lambda_{1}\lambda_{3})(1-\lambda_{2})  < 	(\lambda_{3}-\lambda_{1}\lambda_{3})(1-\lambda_{4})$ is true \\
	for, $ \lambda_{2} > \lambda_{4} $. It would be a contradiction if it were false.\\
	Hence, the condition that needs to be fulfilled to have a successful condition where $ 3^{rd} $ product will be sold with the highest probability is to replace $ 2^{nd}$ ranked product by any other product with lower demand than the $ 2^{nd} $ one.\\
	Empirically this can be shown from the Table in Section (2) for the two possible lists viz. A-B-F and A-D-F. The purchase probability that object F will be bought is higher in the list of offer A-D-F  i.e. $ (1-0.95)(1-0.10)0.75 $ i.e. $ 0.03375 $ compared to A-B-F i.e. $ (1-0.95)(1-0.85)0.75 $ i.e. $ 0.005625 $.
\end{proof}

\textbf{Proposition 3}\textit{	In the presence of information about the customers' fixed attention span, placing the low-ranked product in the middle of two good-ranked items enhances the likelihood of purchasing the next best-ranked products of the included low-ranked product. The criterion is that the two high-scoring goods are arranged in descending order from highest to lowest.  However, this does not guarantee that the expected revenue would be maximized. }
\begin{proof}
	Let there be two lists (with reference from Proposition 2) $ \sigma^{1}=1-2-3 \& \sigma^{2}= 1-4-3 $ and a fixed attention span $ y=3 $. Given these parameters, the expected revenue functions are given below:\\
	The expected revenue from a customer with an attention span of $ y $ can be expressed as
	\begin{equation}\label{10}
		\mathcal{R}(\sigma, y)\triangleq \sum_{k=1}^{y\wedge M} \prod_{i=1}^{k-1}(1-\lambda_{\sigma_{i,n\mid t}}).\lambda_{\sigma_{k,n\mid t}}. p_{\sigma_{k,n\mid t}}.\omega_{k}
	\end{equation}
	For $ \sigma^{1} $ the expected revenue for all is possible ranks and fixed attention span is 
	\begin{equation}\label{key}
		\mathcal{R}(\sigma^{1},y)=	(1-\lambda_{1})(1-\lambda_{2})\lambda_{3}p_{3}\omega_{3}+(1-\lambda_{1})\lambda_{2}p_{2}\omega_{2}+\lambda_{1}p_{1}\omega_{1}.
	\end{equation}
	For $ \sigma^{2} $ the expected revenue for all is possible ranks and fixed attention span is 
	\begin{equation}\label{key}
		\mathcal{R}(\sigma^{2},y)=	(1-\lambda_{1})(1-\lambda_{4})\lambda_{3}p_{3}\omega_{3}+(1-\lambda_{1})\lambda_{4}p_{4}\omega_{4}+\lambda_{1}p_{1}\omega_{1}.
	\end{equation}
	Now consider a contradiction, letting $  \mathcal{R}(\sigma^{2},y)>	\mathcal{R}(\sigma^{1},y)$. This means the first part i.e. $ (1-\lambda_{1})(1-\lambda_{4})\lambda_{3}p_{3}\omega_{3}> (1-\lambda_{1})(1-\lambda_{2})\lambda_{3}p_{3}\omega_{3}$ because $ \lambda_{2}>\lambda_{4} $. The third part is also same i.e.  $ \lambda_{1}p_{1}\omega_{1} $ in both of the cases. And according to the condition $  	\mathcal{R}(\sigma^{2},y)>	\mathcal{R}(\sigma^{1},y)$ the second terms state that $ (1-\lambda_{1})\lambda_{4}p_{4}\omega_{4}> (1-\lambda_{1})\lambda_{2}p_{2}\omega_{2}$. This means $  \lambda_{4}p_{4}\omega_{4}>\lambda_{2}p_{2}\omega_{2}$. This is not true all the time. Because $ \lambda_{4}<\lambda_{2} $. Hence if demand is lower so the price would also be lower i.e. $ p_{4}<p_{2} $. Therefore, for the equal condition i.e. $  \omega_{2}=\omega_{4}$ the this is not true that $  \mathcal{R}(\sigma^{2},y)>	\mathcal{R}(\sigma^{1},y)$. The possibility that $  \omega_{4}=\omega_{2}$ is also lower because higher demand followed by higher price will push the platform to maintain stable $ \omega_{2} $ until product $ 4 $ is owned by the platform and then it would be a colluding rank. This collusion is against antitrust practices.
\end{proof}
	\bibliographystyle{apalike}
	\bibliography{myreferences}
\end{document}